\documentclass[reprint, superscriptaddress, aps, longbibliography, prl]{revtex4-2}

\usepackage{amsmath, amssymb, amsthm, times, mathtools, txfonts, overpic, graphicx, color}
\newcommand{\sslash}{\mathbin{/\mkern-6mu/}}
\newcommand{\K}{\mathbf{K}}
\newcommand{\M}{\mathbf{M}}
\newtheorem*{theorem}{Theorem}

\begin{document}

\title{Mutual Linearity of Complexes in Chemical Reaction Networks}

\author{Pedro E. Harunari}
\email{pedro.harunari@cnrs.fr}
\affiliation{Aix Marseille Université, CNRS, CINAM, Turing Center for Living Systems, 13288 Marseille, France}

\author{Guilherme Fiusa}
\affiliation{Department of Physics and Astronomy, University of Rochester, Rochester, New York 14627, USA}
\affiliation{University of Rochester Center for Coherence and Quantum Science, Rochester, New York 14627, USA}

\author{Matteo Polettini}
\affiliation{Via Leandro Alberti 93, 40139 Bologna, Italy}

\begin{abstract}
Exact statements about how a system far from equilibrium responds to changes in its rate constants are scarce beyond linear response, yet they are central to the control of biochemical networks. A notable exception is mutual linearity: in a jump process, controlling the rates of a single transition makes any two stationary probabilities obey an exact affine relation, arbitrarily far from equilibrium. With nonlinear mass-action kinetics, and conservation laws in place of a normalization, chemical reaction networks seem to defy mutual linearity since species concentrations obey no such relation. We show that it was there all along, in the complexes' activities (mass-action monomials) rather than in the concentrations of species. For zero-deficiency networks with a single linkage class, controlling one reaction makes any three stationary activities obey a linear relation whose coefficients are independent of the controlled rate constants, the conserved quantities (hence the initial state), and even of the stoichiometry of complexes. This predicts hidden activities and bounds the activity ratios the control can reach, while a violation signals nonzero deficiency or multiple linkage classes.
\end{abstract}

\maketitle


Nonequilibrium physics is, to a large extent, a theory of response: how the main observables of a system react to the forces that drive it~\cite{kuboStatisticalMechanicalTheoryIrreversible1957a, marconiFluctuationDissipationResponse2008}. Near equilibrium, response is governed by linear-response theory and the fluctuation--dissipation theorem, which tie the response to equilibrium fluctuations. They do not survive far from equilibrium, where response is instead characterized through few exact expressions and thermodynamic bounds~\cite{aslyamovNonequilibriumResponseMarkov2024, aslyamovGeneralTheoryStatic2024a, fernandesmartinsTopologicallyConstrainedFluctuations2023a, dechantFluctuationResponseInequality2020, ptaszynskiDissipationBoundsPrecision2024, chunTradeoffsNumberFluctuations2023}. Currents and concentrations depend nonlinearly on the control, response coefficients acquire their own dynamics, and the relations that make the near-equilibrium regime tractable become rare and, correspondingly, valuable. Where they do occur, they furnish inference schemes, model-discrimination tests, and benchmarks that hold regardless of how hard the system is driven. The question is also a practical one, to control a biochemical network is to know how its observables respond when some of its rate constants are changed~\cite{governOptimalResourceAllocation2014, qianThermodynamicKineticAnalysis2003, estradaInformationIntegrationEnergy2016, owenUniversalThermodynamicBounds2020, ilker2022shortcuts}. Controlling a single reaction while the rest of the network is left intact is a particularly relevant instance of this: it is how enzymes, inhibitors, and point mutations act. Metabolic control analysis is built on precisely this kind of reaction perturbation~\cite{kacser1973control, heinrich1974linear, Fell1992, reder1988metabolic}. Structural sensitivity analysis showed that the network alone determines which concentrations respond to such a perturbation, but not by how much~\cite{mochizuki2015sensitivity, okada2016law}. A sharper version of the question arises in adaptation, where a system must return an observable to its prestimulus value after a perturbation~\cite{Barkai1997, shinar2010structural}; whether retuning a single reaction can achieve this at all is a question of reachability. Stochastic thermodynamics~\cite{seifertStochasticThermodynamics2012, Shiraishi2023, Seifert2025, Peliti2021} and its extension to chemical reaction networks (CRNs)~\cite{polettiniEspositoIrreversible2014, raoEspositoNonequilibrium2016, avanziniMethodsConversations2024, schmiedl2007stochastic, avanzini2021noneq, tu2026noneq, esposito2020open} have made the search for such relations a central theme.

Answering such questions is hard in general. A change in a single reaction propagates through every cycle it belongs to, so each stationary concentration and current becomes a high-degree rational function of the controlled rate constants, entangled with all the other rate constants of the network. One would expect the responses of different observables to be unrelated, and any exact statement about them to require knowledge of the entire network or a perturbative regime. It therefore came as a surprise that this is not so. Some of us~\cite{harunariMutualLinearityNonequilibrium2024} proved that in a jump process (continuous-time Markov chain), controlling the forward and backward rates of a single edge makes any two stationary currents exactly linear-affine related, arbitrarily far from equilibrium. 
Since then, several works exploring mutual linearities have emerged~\cite{khodabandehlou2025affine, dalcengioMutualMultilinearityNonequilibrium2025b, polettiniCoplanarityRootedSpanningtree2025, zheng2026mutual, zhengMutualLinearityNonequilibrium2026, voitsGraphTheoreticDerivation2026b, floydLocalImperfectFeedback2025, bao2024nonequilibrium}.
In particular, Bebon and Speck~\cite{bebonMutualLinearityGeneric2026a} sharpened this to its most fundamental form: it is the stationary state probabilities themselves that are mutually linear, from which the linearity of currents and of a broad class of observables follows. 
The underlying mechanism is witnessed by applying the Markov chain tree theorem (MCTT)~\cite{schnakenbergNetworkTheory1976}, that expresses stationary probabilities as spanning-tree polynomials of the transition rates, revealing their mutual dependence.

These results belong to jump processes, whose dynamics are linear in the probability vector, and this vector is conveniently normalized. Chemistry is not so accommodating. Under mass-action kinetics~\cite{hornJacksonGeneralMassAction1972, feinbergFoundations2019} the dynamics of species concentrations are genuinely nonlinear and complexes are multimolecular. There is no a priori reason for mutual linearity to persist. Here we show that it does, with mutual linearity arising in complexes' activities. For so, we consider networks with a single linkage class and zero deficiency~\cite{feinbergFoundations2019, marehalliDeficiencyKinetic2023}, a topological--stoichiometric condition met by many biochemical models. Zero deficiency enforces complex balancing~\cite{hornJacksonGeneralMassAction1972, anderson2010product, feinberg1972complexbalancing, horn1972complex}, so the MCTT applies in the space of complexes. We prove that, when a single reaction is controlled, any three stationary complexes' activities obey an exact linear relation. Whereas in jump processes any two stationary probabilities are affinely related~\cite{bebonMutualLinearityGeneric2026a}, in CRNs the constant term, fixed by normalization, is replaced by a third activity:
\begin{align*}
    \text{jump processes:} \quad & p_a = \alpha\, p_b + \beta, \\
    \text{CRNs:} \quad & \psi_a = \alpha\, \psi_b + \beta\, \psi_c,
\end{align*}
where in both cases $\alpha$ and $\beta$ do not depend on the controlled rate constants, which we obtain in closed form. However large the network and however many rate constants it contains, three activities are always enough: two of them generically determine the third, arbitrarily far from equilibrium. The same coefficients also delimit the activity ratios that the control can reach, which settles the reachability question raised above.


\textit{Setup.} A CRN is built from species that react through $n_c$ complexes, the multisets of species entering a reaction as reactants or products. The stoichiometry is carried by the composition matrix $\mathbf{Y}$, whose element $Y_{ic}$ represents the number of the $i$-th species in complex $c$. Reactions are transitions between complexes; we take them all bidirectional, so the complexes are the vertices of an undirected graph whose connected components are the $\ell$ linkage classes.

Under mass-action kinetics the state is the vector of species concentrations $\mathbf{x}$, and each complex carries an activity $\psi_c(\mathbf{x})=\prod_i x_i^{Y_{ic}}$. This activity is the mass-action product of the concentrations of its constituent species, also called the mass-action monomial. We stress that $\psi_c$ is an activity, not a concentration: only for a single-species complex it reduces to a genuine concentration. Collecting them into $\boldsymbol{\psi}$, the dynamics read
\begin{equation}\label{eq:dynamics}
  \dot{\mathbf{x}} = \mathbf{Y}\mathbf{K}\,\boldsymbol{\psi} = \mathbb{S} \boldsymbol j,
\end{equation}
where $\mathbb{S}$ is the stoichiometric matrix, $\boldsymbol j$ are the reaction currents, $\mathbf{K}$ is the Laplacian on the space of complexes, weighted by the rate constants (off-diagonal $K_{cc'}$ is the rate constant of $c'\to c$ and columns sum to zero), i.e.\ exactly the generator of a continuous-time random walk on the complexes.

The interplay of the two ambient spaces, species and complexes, is measured by the deficiency~\cite{feinberg1987chemical}
\begin{equation}\label{eq:deficiency}
  \delta = n_c - \ell - s \ge 0,
\end{equation}
with $s=\operatorname{rank} \mathbb{S}$ the dimension of the stoichiometric subspace~\cite{feinbergFoundations2019, marehalliDeficiencyKinetic2023}. Deficiency is a nonnegative integer fixed by topology and stoichiometry alone, independent of the rate constants. We restrict to $\delta=0$, where $\ker\mathbf{Y}\mathbf{K}=\ker\mathbf{K}$. Every stationary state is then complex-balanced~\cite{hornJacksonGeneralMassAction1972, anderson2010product}, $\mathbf{K}\,\boldsymbol{\psi}=0$, so stationarity in species space, $\dot{\mathbf{x}}=0$, collapses onto the kernel of the complex-space Laplacian.

Assuming irreducibility of the space of complexes, there is a single linkage class, and $\ker\mathbf{K}$ is one-dimensional. Under these conditions, we can apply the MCTT~\cite{schnakenbergNetworkTheory1976}, which holds for any Laplacian on a connected graph~\cite{Mirzaev2013Laplacian}, and obtain the direction of the stationary state as
\begin{equation}\label{eq:MCTT}
  \boldsymbol{\psi} \propto \boldsymbol{\tau}, \qquad \tau_c = \sum_{\mathcal{T}_c}\omega(\mathcal{T}_c),
\end{equation}
where $\tau_c$ is the sum, over all spanning trees $\mathcal{T}_c$ rooted at $c$ (edges oriented towards $c$), of the weights $\omega(\mathcal{T}_c)$, each the product of the rate constants along the tree. The actual activity of each complex can be expressed as $\psi_c = \tau_c / Z$, where $Z$ is a common denominator to all complexes set by conservation laws and the initial state. In a jump process, normalization would give $Z=\sum_c \tau_c$; in a CRN it does not. For instance, $2X \xrightleftharpoons[k_-]{k_+} Y$ gives $(x^2, y) = (k_-,k_+)/Z$, which leaves $Z$ undetermined until the conserved quantity $x+2y$, fixed by the initial state, selects it, whereas a jump process with the same rates would have $Z=k_-+k_+$. The stationary complexes' activities are thus proportional to spanning-tree polynomials, the very object behind mutual linearity in jump processes, now living one level up, in the space of complexes.

Finally, we control a single reaction, considering its forward and backward rate constants $r_\pm$ as free parameters. Importantly, we do not assume a mere perturbation of their values nor restrict to near equilibrium conditions. We require only that the controlled reaction is bidirectional and not a bridge, i.e. its removal must leave the space of complexes connected.


\textit{Mutual linearity.} Fix the controlled reaction and let $r_\pm$ be its two directed rate constants. In a rooted spanning tree each edge is either absent or traversed in a single direction, so the controlled reaction contributes to $\omega(\mathcal{T}_c)$ either $r_+$, or $r_-$, or nothing at all: the tree polynomials $\boldsymbol{\tau}$ are at most linear in $r_\pm$ and never contain both. This splits the tree vector into three,
\begin{equation}\label{eq:tau_decomp}
  \boldsymbol{\tau}(r_\pm) = \boldsymbol{\tau}^{\backslash r_\pm} + \boldsymbol{\tau}^{\backslash r_- \sslash r_+} (r_+) + \boldsymbol{\tau}^{\backslash r_+ \sslash r_-} (r_-),
\end{equation}
collecting respectively the trees that avoid the controlled reaction, those that avoid $r_-$ but pass through $r_+$, and those that avoid $r_+$ but pass through $r_-$.

Our proof relies on the following: with tree surgery, Ref.~\cite{polettiniCoplanarityRootedSpanningtree2025} showed that the vectors $( \tau_c^{\backslash r_\pm}, \tau_c^{\backslash r_- \sslash r_+}, \tau_c^{\backslash r_+ \sslash r_-} )$, one for each complex $c$, are coplanar whenever the controlled reaction is not a bridge. Intuitively, removing the controlled reaction leaves a connected network whose stationary state is $\boldsymbol\tau^{\backslash r_\pm}$. On this network, the controlled reaction acts as a source--sink pair, injecting a current at one endpoint and withdrawing it at the other, and the trees through $r_+$ and through $r_-$ describe the responses to injections in opposite directions. In this picture, coplanarity states that a suitably weighted combination of the two responses, in which the opposite injections balance exactly, is the stationary state of the uncontrolled network.

\begin{theorem}[Mutual linearity of complexes' activities]
Consider a mass-action CRN with bidirectional reactions, zero deficiency and a single linkage class, and let $r_\pm$ be the rate constants of a reaction that is not a bridge, i.e.\ whose removal leaves the space of complexes connected. Then, for any three complexes $a$, $b$, $c$ and all values of $r_\pm$, the positive stationary activities satisfy
\begin{equation}\label{eq:mutuallinearity}
    \Delta_{bc} \psi_a + \Delta_{ca} \psi_b + \Delta_{ab} \psi_c = 0,
\end{equation}
with coefficients $\Delta$ (defined in the proof) that depend on neither $r_\pm$, the composition matrix $\mathbf Y$, the conservation laws, nor the initial state.
\end{theorem}

\begin{proof}
By coplanarity, the matrix $\M \coloneqq ( \boldsymbol \tau^{\backslash r_\pm}, \boldsymbol\tau^{\backslash r_- \sslash r_+}, \boldsymbol\tau^{\backslash r_+ \sslash r_-} )$ has rank at most two, and so does the square submatrix $\M_{abc}$ collecting the rows of $a$, $b$, $c$; thus $\det \M_{abc} = 0$. From $\mathrm{adj} (\M_{abc}) \M_{abc} = \det (\M_{abc}) \mathbf I = \mathbf 0$, every row of the adjugate is orthogonal to the three columns of $\M_{abc}$. The third row collects the cofactors along the column carrying $\boldsymbol\tau^{\backslash r_+\sslash r_-}$, which therefore drops out; divided by $r_+$, it is $\boldsymbol \lambda \coloneqq ( \Delta_{bc}, \Delta_{ca}, \Delta_{ab} )$, with $\Delta_{ab} = ( \tau_a^{\backslash r_\pm} \tau_b^{\backslash r_- \sslash r_+} - \tau_b^{\backslash r_\pm} \tau_a^{\backslash r_- \sslash r_+}) /r_+ = -\Delta_{ba}$ independent of $r_\pm$. Being orthogonal to each column, $\boldsymbol\lambda$ is orthogonal to their sum $(\tau_a, \tau_b, \tau_c)$, hence by Eq.~\eqref{eq:MCTT} to $(\psi_a, \psi_b, \psi_c)$, which is Eq.~\eqref{eq:mutuallinearity}. Finally, $\mathbf Y$, the conservation laws and the initial state enter $\boldsymbol\psi$ only through the common factor $Z$, which cancels from this homogeneous relation.
\end{proof}

Every three activities thus trace a plane as the controlled reaction changes, and the content of the Theorem is graded by its coefficients: \textit{(i)} when no coefficient vanishes, the mutual linearity is irreducibly three-term and two activities determine the third; \textit{(ii)} when one does, say $\Delta_{ab}$, the pair it joins is rigidly proportional, $\psi_a/\psi_b = \tau_a^{\backslash r_\pm}/\tau_b^{\backslash r_\pm}$ for all $r_\pm$, and the third complex decouples from them. Geometrically, $\boldsymbol\lambda$ is the normal to the plane confining $(\psi_a,\psi_b,\psi_c)$; a vanishing $\Delta_{ab}$ makes this plane contain the $\psi_c$ axis, so that $(\psi_a, \psi_b)$ stays on a fixed line while $\psi_c$ is unconstrained.


\textit{Generic linear combinations of the complexes' activities.} Let $\mathcal O_i \coloneqq \boldsymbol o_i \cdot \boldsymbol \psi$ be any observable that is linear in the activities, and whose weights $\boldsymbol o_i$ do not depend on the controlled rate constants $r_\pm$. The product of a matrix whose rows are $\boldsymbol o_1$, $\boldsymbol o_2$ and $\boldsymbol o_3$, with the one whose columns are $\boldsymbol \tau^{\backslash r_\pm}$, $\boldsymbol \tau^{\backslash r_- \sslash r_+}$ and $\boldsymbol \tau^{\backslash r_+ \sslash r_-}$ has rank at most 2, since the second matrix has rank at most 2 from the arguments above. Following the same procedure of taking the third row of the adjugate matrix, which will be orthogonal to the vectors composing each observable, we get
\begin{equation}\label{eq:observable}
    \Delta(\boldsymbol o_2, \boldsymbol o_3) \mathcal O_1 + \Delta(\boldsymbol o_3, \boldsymbol o_1) \mathcal O_2 + \Delta(\boldsymbol o_1, \boldsymbol o_2) \mathcal O_3 = 0 ,
\end{equation}
where $\Delta(\boldsymbol o_1, \boldsymbol o_2) \coloneqq \sum_{c, c'} o_1(c) o_2(c') \Delta_{cc'} = - \Delta(\boldsymbol o_2, \boldsymbol o_1)$ are independent of the controlled rate constants $r_\pm$. This equation shows that such observables satisfy a mutual linearity relation.

The Theorem is the special case where observables are defined by the canonical vectors $\boldsymbol o_1, \boldsymbol o_2, \boldsymbol o_3 = \mathbf e_a, \mathbf e_b, \mathbf e_c$. Many observables fall in this class, such as weighted sums of complexes' activities, dynamical activity of some reactions, fluxes to reservoirs and their chemical power. Also, reaction currents satisfy this relation: the reaction current $j_{m\to n}=r_{m\to n}\psi_m - r_{n\to m}\psi_n$ can be expressed as this type of observable for $\boldsymbol o = r_{m\to n} \mathbf e_m - r_{n\to m} \mathbf e_n$; therefore, for any three currents,
\begin{equation}\label{eq:currentlinearity}
  \Delta(\boldsymbol o_2, \boldsymbol o_3) j_1 + \Delta(\boldsymbol o_3, \boldsymbol o_1) j_2 + \Delta(\boldsymbol o_1, \boldsymbol o_2) j_3 = 0.
\end{equation}
Equation~\eqref{eq:currentlinearity} applies directly only to uncontrolled reactions: the current of the controlled reaction, $j_r = r_+\psi_m - r_-\psi_n$ with $m$ and $n$ its endpoints, has weights $\boldsymbol o_r = r_+\mathbf e_m - r_-\mathbf e_n$ that depend on $r_\pm$. It can nevertheless be included. Complex balancing at $m$, $(\K\boldsymbol\psi)_m = 0$, states that the currents into $m$ sum to zero, so $j_r$ equals the sum of the currents of the other reactions incident on $m$, each oriented towards $m$. Such reactions exist, since a reaction that is not a bridge lies on a cycle. Replacing $\boldsymbol o_r$ by the corresponding sum of uncontrolled weights, the controlled current also obeys Eq.~\eqref{eq:currentlinearity} and may serve as one of the three. Notice that the same trick does not work for observables other than currents. For example, exactly as in jump processes~\cite{bebonMutualLinearityGeneric2026a}, the full dynamical activity retains a genuine dependence on the control due to $r_+\psi_m + r_-\psi_n$. Some thermodynamic quantities such as affinities and entropy production are not linear in the activities, therefore they do not follow this relation, reinforcing that mutual linearity is a purely kinetic fact.

Operationally, two measurements of any two such observables determine the third, however unrelated they look: an activity, a one-way flux and a chemical power are constrained to a common plane through the origin as the control is swept. Take the example of a unicyclic network, and consider two currents $j_1$ and $j_2$, and one activity $\psi_c$. It can be shown that $ \Delta( \boldsymbol o_1, \boldsymbol o_2) = 0$ and $\Delta ( \mathbf e_c, \boldsymbol o_1) = \Delta ( \mathbf e_c, \boldsymbol o_2)$. Substituting these in Eq.~\eqref{eq:observable}, we get $j_1 = j_2$ for any value of the controlled rate constants, which is a known consequence of its topology.

\textit{Stochastic dynamics.} So far the discussion has focused on deterministic mass-action kinetics. The conclusion also holds for stochastic mass-action kinetics. By the Anderson-Craciun-Kurtz theorem~\cite{anderson2010product}, for zero-deficiency networks the stationary probability for the species is a product-form Poisson-like distribution. This entails that the mean stochastic stationary currents are identical to the deterministic ones~\cite{polettini2015dissipation}, so they obey Eq.~\eqref{eq:currentlinearity} with the same coefficients.

\textit{Chemostatted species.} In open networks, a species is exchanged with a reservoir through the null complex $\varnothing$, which carries unit activity $\psi_\varnothing=1$~\cite{Cengio2022}. The proof of the Theorem holds by replacing $c$ with $\varnothing$, and due to the vanishing determinant of matrix $\M_{ab\varnothing}$, we obtain
\begin{equation}\label{eq:affine}
	\Delta_{b\varnothing} \psi_a+ \Delta_{\varnothing a} \psi_b + \Delta_{ab} = 0,
\end{equation}
a linear affine relation between any two activities, meaning that measuring a single reference activity $\psi_b$ suffices to determine $\psi_a$ for any complex $a$ such that $\Delta_{b\varnothing} \neq 0$.

\textit{Jump processes.} When every complex is a single species, $\mathbf{Y}$ is the identity and Eq.~\eqref{eq:dynamics} is the master equation of a jump process, whose activities $\psi_c \to p_c$ are state probabilities. The algebra behind Eq.~\eqref{eq:observable} is unchanged, and jump processes have an additional ingredient: normalization, $\mathbf 1 \cdot \boldsymbol p = 1$. This is itself an observable with rate-independent weights, so we may choose $\boldsymbol o_3 = \mathbf 1$. With $\boldsymbol o_1, \boldsymbol o_2 = \mathbf e_a, \mathbf e_b$, Eq.~\eqref{eq:observable} becomes
\begin{equation}\label{eq:markov}
    \Delta_{b\mathbf 1} p_a + \Delta_{\mathbf 1 a} p_b + \Delta_{ab} = 0, \qquad \Delta_{a \mathbf 1} \coloneqq \sum_c \Delta_{ac} = -\Delta_{\mathbf 1 a},
\end{equation}
an affine relation between any two state probabilities, which recovers the mutual linearity of Ref.~\cite{bebonMutualLinearityGeneric2026a}. Using this result, it also recovers the mutual linearity of currents of Ref.~\cite{harunariMutualLinearityNonequilibrium2024}. Normalization is typically absent in CRNs, and it is what reduces the three-term law to a two-term affine one.

\textit{Bounds on activity ratios.} Choosing $c = m$, the source of $r_+$, in Eq.~\eqref{eq:mutuallinearity} and dividing by $\Delta_{bm}\psi_b$ gives
\begin{equation}\label{eq:ratio}
    \frac{\psi_a}{\psi_b} = \frac{\Delta_{ma}}{\Delta_{mb}} + \frac{\Delta_{ab}}{\Delta_{mb}} \frac{\psi_m}{\psi_b},
\end{equation}
and similarly for $c=n$, the source of $r_-$. Since $\psi_m/\psi_b \ge 0$, $\Delta_{ma}/\Delta_{mb}$ is either a lower or upper bound for $\psi_a / \psi_b$ depending on the sign of $\Delta_{ab} / \Delta_{mb}$. It can be shown that $\Delta_{mb} \geq 0$ and $\Delta_{nb} \leq 0$~\footnote{A tree rooted at $m$ cannot contain $r_+$, so $\tau^{\backslash r_- \sslash r_+}_m = 0$ and $\Delta_{mb} = \tau^{\backslash r_\pm}_m \tau^{\backslash r_- \sslash r_+}_b / r_+ \geq 0$. Likewise $\tau^{\backslash r_+ \sslash r_-}_n = 0$. Writing the coplanarity relation as $\boldsymbol\tau^{\backslash r_\pm} = x \boldsymbol\tau^{\backslash r_+ \sslash r_-} + y \boldsymbol\tau^{\backslash r_- \sslash r_+}$, its $m$ component gives $x = \tau^{\backslash r_\pm}_m / \tau^{\backslash r_+ \sslash r_-}_m > 0$, and its $n$ component gives $\Delta_{nb} = -x  \tau^{\backslash r_+ \sslash r_-}_b \tau^{\backslash r_- \sslash r_+}_n / r_+ \leq 0$}, therefore
\begin{equation}\label{eq:bounds}
    \min \left\lbrace \frac{\Delta_{ma}}{\Delta_{mb}}, \frac{\Delta_{na}}{\Delta_{nb}} \right\rbrace \le \frac{\psi_a}{\psi_b} \le \max \left\lbrace \frac{\Delta_{ma}}{\Delta_{mb}}, \frac{\Delta_{na}}{\Delta_{nb}} \right\rbrace ,
\end{equation}
independently of $r_\pm$. The bounds are tight: $r_+\to\infty$ drains the source $m$, so $\psi_m/\psi_b \to 0$ and the ratio tends to $\Delta_{ma}/\Delta_{mb}$, while $r_-\to\infty$ drains $n$ and the ratio goes to $\Delta_{na}/\Delta_{nb}$. The sign of $\Delta_{ab}$ decides which bound is the lower one, but if $\Delta_{ab}=0$ the interval collapses onto a fixed ratio and the third complex decouples. The bounds thus follow from mutual linearity itself, and delimit the values that the ratio of two activities can take as the controlled reaction is driven over its whole range. Since the ratio is continuous in $r_\pm$ and approaches both bounds, every value strictly between them is attainable. This answers the reachability question of the introduction: an activity ratio can be brought to a target value by retuning the controlled reaction alone if and only if the target lies inside this interval.

Reference~\cite{liangThermodynamicSpaceChemical2026} delimits the stationary concentrations accessible to a network on a given thermodynamic budget, including activity ratios at fixed thermodynamic driving. Meanwhile, Eq.~\eqref{eq:bounds} holds for fixed kinetics apart from the single controlled reaction, which can change thermodynamic quantities. The two meet at constant $r_+/r_-$, which preserves the driving force. For jump processes, Ref.~\cite{bebonMutualLinearityGeneric2026a} provides asymptotic bounds on stationary probabilities that come from the constriction ensembles of Eq.~\eqref{eq:tau_decomp}; in this case, normalization is what allows each probability to be bounded on its own.


\textit{A textbook case.} Consider the reversible Michaelis--Menten mechanism, closed into a cycle by an enzyme-catalysed conversion:
\begin{equation}\label{eq:MM}
  E+S \xrightleftharpoons[k_{-1}]{k_{+1}} ES \xrightleftharpoons[k_{-2}]{k_{+2}} E+P \xrightleftharpoons[k_{-3}]{k_{+3}} E+S .
\end{equation}
The complexes form a single linkage class, all reactions are bidirectional, and the deficiency vanishes. Crucially, two of them are genuinely bimolecular, so $\mathbf{Y}$ is not the identity and the network is not a Markov chain in disguise, it carries genuine CRN nonlinearity features. The first reaction $E+S\rightleftharpoons ES$, representing the binding of enzyme and substrate, is not a bridge: deleting it keeps all complexes connected. All hypotheses are thus met, and by controlling the rate constants of one reaction, mutual linearity provides an exact relation for the stationary activities of all three complexes (the monomials $[E][S]$, $[ES]$, and $[E][P]$) with coefficients independent of the controlled rate constants. The statement is unchanged when further intermediates are resolved, as in the four-complex cycle $E+S\rightleftharpoons ES \rightleftharpoons EP \rightleftharpoons E+P\rightleftharpoons E+S$, which preserves zero deficiency.

\begin{figure}
    \raisebox{-\height}{%
  \begin{overpic}[width=.24\textwidth, clip, trim=5 5 0 0]{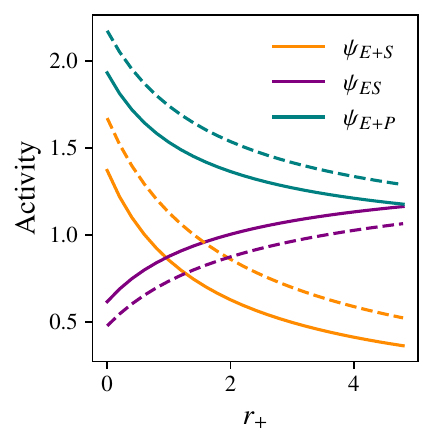}
      \put(0,95){(a)}
  \end{overpic}%
}%
\raisebox{-\height}{%
  \begin{overpic}[width=.24\textwidth, clip, trim=5 5 0 0]{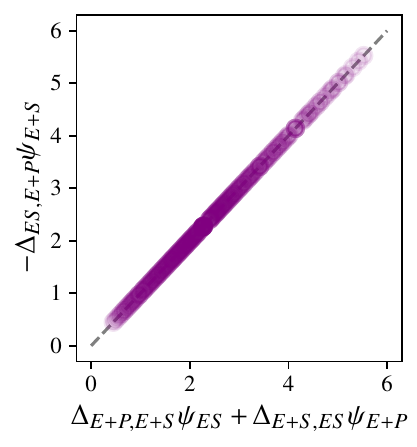}
      \put(0,95){(b)}
  \end{overpic}%
}
    \caption{\textit{Reversible Michaelis-Menten.} (a) Activities of all complexes as function of the perturbed rate constant $r_+$, with $r_- = 1$ in solid lines and 2 in dashed lines. (b) Products of $\Delta$s and $\psi$s that should be equal due to mutual linearity, sampled in a 2D grid of $r_\pm \in [0,5]$; the dashed line is $y=x$. The initial concentration of all species is 1, $( k_{+2}, k_{-2}, k_{+3}, k_{-3}) = (1.2, 0.7, 1.1, 2)$, and $k_{\pm 1} = r_\pm$ are controlled.}
    \label{fig:MM}
\end{figure}

Figure~\ref{fig:MM} illustrates the Theorem on the reversible Michaelis--Menten cycle. Controlling the first reaction, we sweep its rate constants $k_\pm = r_\pm$ into the regime where the three complexes' activities respond non-linearly, cf. Fig.~\ref{fig:MM}(a). If we plot $-\Delta_{bc} \psi_a$ versus $\Delta_{ca}\psi_b + \Delta_{ab} \psi_c$ it collapses to the line $y=x$ across the whole two-parameter family to machine precision, illustrating mutual linearity, cf. Fig.~\ref{fig:MM}(b).

\textit{A richer space of complexes.} Consider a nonlinear network with complexes $2A,2B,2C,2D,2E$ joined in a single linkage class with three independent cycles,
\begin{eqnarray}\label{eq:toy}
    2A \rightleftharpoons2B \rightleftharpoons 2C \rightleftharpoons 2D \rightleftharpoons 2E \rightleftharpoons 2A, \nonumber\\
    2A \rightleftharpoons 2C, \quad 2B \rightleftharpoons 2D .
\end{eqnarray}
This example has zero deficiency and is irreducible, satisfying the conditions for mutual linearity of activities. Every complex is bimolecular, $\psi_{2X}=[X]^2$, so the kinetics is nonlinear. We verify in Fig.~\ref{fig:2A2B2C2D2E} the mutual linearity relation for complexes $\{ 2A, 2B,2C\}$ and $\{ 2A, 2B, 2E\}$, when controlling $2C \rightleftharpoons 2D$. The bound of Eq.~\eqref{eq:bounds} is represented by stars, and is verified. One interesting aspect of this example is that the composition matrix $\mathbf Y$ is twice the identity, so it is invertible, and the mutual linearity of activities directly implies a quadratic equation for every species concentration: $\Delta_{bc} [A]^2 + \Delta_{ca} [B]^2 + \Delta_{ab} [C]^2 = 0$, and similarly for any other permutation of species. In the inset of Fig.~\ref{fig:2A2B2C2D2E}, we take the same CRN but replace $2E$ by $A+B$, which decreases the rank of the stoichiometric matrix by one, and the deficiency is $\delta = 1$; in this case, the relation is violated. Since this network has a single linkage class, the violation can only be due to the nonzero deficiency, and mutual linearity thus doubles as a data-driven test of deficiency.

\begin{figure}
    \centering
    \includegraphics[width=\linewidth]{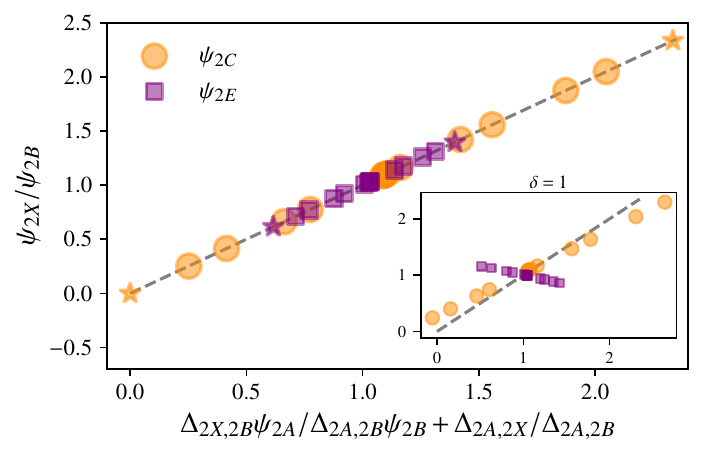}
    \caption{Mutual linearity witnessed in the ratio of $2C$ (orange circles) and $2E$ (purple squares) with respect to $2B$. The dashed gray line is $y=x$, and the stars represent the bounds from Eq.~\eqref{eq:bounds}. The rate constants of $2C \rightleftharpoons 2D$ are swept from 0 to 10; the initial concentrations are respectively $\{1, 2, 3, 1.1, 2.4 \}$; the rate constants are $k_{\pm 1} = \{ 1.2, 1 \}$, $k_{\pm 2} = \{ 1.3, 1.1 \}$, $k_{\pm 4} = \{ 1.5, 1.3 \}$, $k_{\pm 5} = \{ 1.6, 1.4 \}$, $k_{\pm 6} = \{ 1.7, 1.5 \}$, and $k_{\pm 7} = \{ 1.8, 1.6 \}$. The inset has complex $2E$ replaced by $A+B$: same reaction topology, same rate constants, but deficiency $\delta = 1$.}
    \label{fig:2A2B2C2D2E}
\end{figure}


\textit{Outlook.} Because the coefficients $\Delta$ are independent of the control, mutual linearity is a ready tool for inference and model discrimination, and the coefficients can be obtained in two ways. They can be calculated analytically from the spanning-tree polynomials, but this can become a daunting task: the number of spanning trees typically grows exponentially with network size, some reactions may be hidden, or one might simply not be interested in learning about conditioned rooted spanning-tree polynomials. In such cases, they can be estimated empirically. Since Eq.~\eqref{eq:mutuallinearity} is homogeneous, two triples $(\psi_a,\psi_b,\psi_c)$ measured at different control settings fix the coefficients up to a common factor, and any further triple tests the relation. Either way, once the coefficients are known, $\psi_a$ and $\psi_b$ predict $\psi_c$, and a measured departure from coplanarity signals nonzero deficiency or multiple linkage classes. Deficiency, a structural invariant normally read off the full network, thereby becomes testable from steady-state measurements of activities. Connecting these coefficients to the thermodynamic cost of driving~\cite{polettiniEspositoIrreversible2014, raoEspositoNonequilibrium2016}, and to the deficiency-controlled phenomena of catalysis and kinetic invertibility~\cite{marehalliDeficiencyKinetic2023}, is a natural next step. Likewise, other forms of linearity might hold for the nonzero deficiency case, and we leave it for future contributions.

Although mutual linearity in CRNs is a statement about complexes' activities, it also constrains the underlying species. Since $\psi_c = \prod_i x_i^{Y_{ic}}$, Eq.~\eqref{eq:mutuallinearity} is a polynomial equation in $\mathbf x$ with coefficients independent of the control. As the controlled rate constants are swept, the stationary concentrations therefore remain on the solutions of these equations that are compatible with the conservation laws, which add linear constraints fixed by the initial state. Polynomial equations combined with linear constraints are hard to characterize in general, so a systematic account of what mutual linearity implies for species concentrations is an interesting direction.

\begin{acknowledgments}
{\emph{Acknowledgments}.} We thank Sara Dal Cengio, Vivien Lecomte, Shesha Gopal Marehalli Srinivas, and Efe Ilker for useful discussions. PH acknowledges funding from the European Union (ERC-SuperStoc-101117322). GF was supported by the U.S. Department of Energy (DOE), Office of Science, Basic Energy Sciences (BES) under Award No. DE-SC0025516.

{\emph{AI Statement}.} Claude (Anthropic, Opus 5 and 5.5) was used, under the authors' interactive direction, for aid with exploration suggestions, verifying results, numerical computation, and text editing. All results were independently obtained by the authors, who take full responsibility for the content.
\end{acknowledgments}

\bibliography{refs}

\end{document}